\documentclass[10pt,twocolumn,twoside]{IEEEtran}

\IEEEoverridecommandlockouts  

\usepackage{amssymb,amsmath,amsfonts,amsthm}
\usepackage{algorithm,algorithmic}
\usepackage{cite}
\usepackage{float}
\usepackage{epsfig}
\usepackage{graphicx}
\usepackage{graphics}
\usepackage{subfigure}
\usepackage{url}
\usepackage{xspace}
\usepackage[usenames,dvipsnames]{color}
\usepackage{soul}
\usepackage{mathtools}
\usepackage{tikz}
\usetikzlibrary{arrows,quotes}

\floatstyle{ruled}
\newfloat{model}{H}{mod}
\floatname{model}{\footnotesize Model}
\newfloat{notatio}{H}{not}
\floatname{notatio}{\footnotesize Notation}

\newenvironment{varalgorithm}[1]
  {\algorithm\renewcommand{\thealgorithm}{#1}}
  {\endalgorithm}

\newenvironment{list4}{
	\begin{list}{$\bullet$}{%
			\setlength{\itemsep}{0.05cm}
			\setlength{\labelsep}{0.2cm}
			\setlength{\labelwidth}{0.3cm}
			\setlength{\parsep}{0in} 
			\setlength{\parskip}{0in}
			\setlength{\topsep}{0in} 
			\setlength{\partopsep}{0in}
			\setlength{\leftmargin}{0.16in}}}
	{\end{list}}

\newenvironment{list4a}{
	\begin{list}{$\bullet$}{%
			\setlength{\itemsep}{0.05cm}
			\setlength{\labelsep}{0.2cm}
			\setlength{\labelwidth}{0.3cm}
			\setlength{\parsep}{0in} 
			\setlength{\parskip}{0in}
			\setlength{\topsep}{0in} 
			\setlength{\partopsep}{0in}
			\setlength{\leftmargin}{0.16in}}}
	{\end{list}}

\usepackage{url,changebar,bm,xspace,dsfont}
\let\mathbb=\mathds 
\def\day{\mathrm{day}} 

\newtheorem{theorem}{Theorem}

\newtheorem{prop}{Proposition}
\newtheorem{assum}{Assumption}

\newtheorem{remark}{Remark}

\ifCLASSINFOpdf
 \else
\fi

\usepackage{amsmath,amsfonts,amssymb,amscd}
\usepackage{capt-of}
\usepackage{color}
\usepackage{cite}
\usepackage{verbatim}
\usepackage{hyperref}

\hypersetup{
    colorlinks = true,
    linkcolor = [rgb]{0,0,1},
    anchorcolor = [rgb]{0,0,1},
    citecolor = [rgb]{0.9,0.5,0},
    filecolor = [rgb]{0,0,1},
    urlcolor = [rgb]{0.9,0.5,0},
    bookmarksopen = true,
    bookmarksnumbered = true,
    breaklinks = true,
    linktocpage,
    colorlinks = true,
    linkcolor = [rgb]{0.2,0.6,0.2},
    urlcolor  = [rgb]{0,0,1},
    citecolor = [rgb]{0.9,0.5,0},
    anchorcolor = [rgb]{0.2,0.6,0.2},
}

\begin{document}

\title{\LARGE \bf Distributed Optimization with Finite Bit Adaptive Quantization for Efficient Communication and Precision Enhancement} 

\author{Apostolos~I.~Rikos, Wei~Jiang, Themistoklis~Charalambous, and Karl~H.~Johansson
\thanks{Apostolos~I.~Rikos is with the Artificial Intelligence Thrust of the Information Hub, The Hong Kong University of Science and Technology (Guangzhou), Guangzhou, China. 
He is also affiliated with the Department of Computer Science and Engineering, The Hong Kong University of Science and Technology, Clear Water Bay, Hong Kong, China. 
E-mail: {\tt~apostolosr@hkust-gz.edu.cn}.} 
\thanks{Wei Jiang was with the Department of Electrical Engineering and Automation, School of Electrical Engineering, Aalto University, Espoo, Finland.
Email: {\tt wjiang.lab@gmail.com.}} 
\thanks{Themistoklis~Charalambous is with the Department of Electrical and Computer Engineering, School of Engineering, University of Cyprus, 1678 Nicosia, Cyprus.  E-mail:{\tt~charalambous.themistoklis@ucy.ac.cy}. T.~Charalambous is also with the Department of Electrical Engineering and Automation, School of Electrical Engineering, Aalto University, Espoo, Finland. Email: {\tt themistoklis.charalambous@aalto.fi.}} 
\thanks{Karl~H.~Johansson is with the Division of Decision and Control Systems, KTH Royal Institute of Technology, SE-100 44 Stockholm, Sweden, also affiliated with Digital Futures. E-mails: {\tt kallej@kth.se}.} 
\thanks{Part of this work was supported by the Knut and Alice Wallenberg Foundation, the Swedish Research Council, and the Swedish Foundation for Strategic Research. 
The work of T. Charalambous was partly supported by the European Research Council (ERC) Consolidator Grant MINERVA (Grant agreement No. 101044629).
}
}

\maketitle
\pagestyle{empty}

%
%
%
%
\begin{abstract}
In realistic distributed optimization scenarios, individual nodes possess only partial information and communicate over bandwidth constrained channels. 
For this reason, the development of efficient distributed algorithms is essential. 
In our paper we addresses the challenge of unconstrained distributed optimization. 
In our scenario each node's local function exhibits strong convexity with Lipschitz continuous gradients. 
The exchange of information between nodes occurs through $3$-bit bandwidth-limited channels (i.e., nodes exchange messages represented by a only $3$-bits). 
Our proposed algorithm respects the network's bandwidth constraints by leveraging zoom-in and zoom-out operations to adjust quantizer parameters dynamically. 
We show that during our algorithm's operation nodes are able to converge to the exact optimal solution. 
Furthermore, we show that our algorithm achieves a linear convergence rate to the optimal solution. 
We conclude the paper with simulations that highlight our algorithm's unique characteristics.  
\end{abstract}


%
%
%
%
\section{Introduction}\label{sec:intro}

Distributed optimization has become increasingly important in various domains. 
Such domains range from large-scale machine learning \cite{2020:Nedich}, control \cite{SEYBOTH:2013}, to other data-driven applications \cite{2018:Stich_Jaggi} that involve large amounts of data. 
In the distributed setting multiple nodes collaborate to collectively solve optimization problems with each node possessing only partial information about the global objective. 
Distributing the optimization process offers advantages in terms of scalability, fault tolerance, and adaptability to dynamic environments \cite{nedic2018distributed}.

In current literature, numerous algorithms have been developed for distributed optimization \cite{2009:Nedic_Optim, 2018:Khan_AB, 2020:Vivek_Salapaka, 2021:Tiancheng_Uribe}. 
These algorithms typically necessitate nodes to exchange real-valued messages with high or infinite precision, representing local gradients or optimization variables. 
However, exchanging messages with infinite precision is often impractical or even impossible due to bandwidth constraints, as it requires the transmission of an infinite number of bits per second. 
Consequently, such algorithms tend to incur significant communication overhead, particularly in scenarios involving bandwidth-limited channels.

Due to the aforementioned limitations, there is a growing interest for distributed optimization algorithms that can operate efficiently under communication constraints. 
Specifically, there has been a surge of interest in distributed optimization approaches where nodes exchange compressed packets or quantized messages \cite{2014:Peng_Yiguan, 2016:Huang_Xiao, 2017:Ye_Zeilinger_Jones, 2018:Huaqing_Xie, 2021:Doan_Romberg, 2020:Magnusson_NaLi, 2021:Jhunjhunwala_Eldar, 2021:Kajiyama_Takai, 2022:Liu_Daniel,  2019:Koloskova_Jaggi, 2019:Basu_Diggavi, 2020:Li_Chi, horvath2023stochastic}. 
These approaches aim to reduce communication overhead while still enabling nodes to converge to the optimal solution or approximate it with high accuracy that depends on the utilized quantization level. 
However, in scenarios (such as resource-constrained networks or low-power devices) where communication among nodes is limited to only a small fixed amount of bits (e.g., $3$, $4$, or $5$-bits per message), the aforementioned approaches face limitations. 
Algorithms that respect this constraint usually converge to a neighborhood of the optimal solution leading to an error floor that may be considerably large \cite{2023:Rikos_Johan_IFAC}. 
Refining the quantization level allows for approximation with higher precision (potentially enabling convergence to the exact optimal solution) \cite{rikos2023distributed}. 
But the refinement process increases the number of exchanged bits, leading to constraint violation and may compromise the efficiency gains achieved through quantization. 

\textbf{Main Contributions.} 
Motivated by the aforementioned limitations and challenges, in our paper we propose a distributed optimization algorithm that exhibits efficient communication among nodes. 
The main advantage of our algorithm compared to the current literature is its ability to operate in scenarios where communication is limited to only a small fixed amount of bits. 
Additionally, due to its ability to adjust the quantizer parameters, it enables nodes to converge to the exact optimal solution. 
Our main contributions are the following: 
\begin{itemize}
    \item We propose a distributed optimization algorithm that exhibits efficient communication enabling nodes to exchange messages consisting only $3$-bits; see Algorithm~\ref{alg1}. 
    Our algorithm's operation relies on zoom-in and zoom-out operations according to a set of conditions. 
    These operations adjust the quantizer parameters (such as basis and range) to focus on the estimated optimal solution while ensuring approximation with higher precision. 
    \item We show that our algorithm's operation, enables nodes to estimate the \textit{exact} optimal solution exhibiting linear convergence rate (thus avoiding convergence to a neighborhood of the optimal solution due to the quantization effect); see Theorem~\ref{converge_Alg1}. 
\end{itemize}

\section{NOTATION AND BACKGROUND}\label{sec:preliminaries}

\textbf{Mathematical Notation and Symbols.}  
The symbols $\mathbb{R}, \mathbb{Q}, \mathbb{Z}$, and $\mathbb{N}$ represent the sets of real, rational, integer, and natural numbers, respectively. 
The symbols $\mathbb{Q}_{> 1}, \mathbb{R}_{\geq 0}$ denote the set of rational numbers greater than one, and the set of positive real numbers, respectively. 
The symbol $\mathbb{R}_{\geq 0}$ indicates the nonnegative orthant of the $n$-dimensional real space $\mathbb{R}^n$. 
Matrices are denoted by capital letters (e.g.,  $A$), and vectors by small letters (e.g.,  $x$). 
The transpose of a matrix $A$ and a vector $x$ are denoted as $A^\top$, $x^\top$, respectively.
For any real number $a \in \mathbb{R}$, $\lfloor a \rfloor$ denotes the greatest integer less than or equal to $a$, and $\lceil a \rceil$ denotes the least integer greater than or equal to $a$. 
For a matrix $A \in \mathbb{R}^{n \times n}$, $a_{ij}$ denotes the entry in row $i$ and column $j$. 
The all-ones vector is denoted by $\mathbb{1}$, and the identity matrix by $\mathbb{I}$ of appropriate dimensions. 
The Euclidean norm of a vector is represented by $\| \cdot \|$. 

\textbf{Graph-Theoretic Notions.} 
The communication network is modeled as a directed graph (digraph), denoted by $\mathcal{G} = (\mathcal{V}, \mathcal{E})$. 
This digraph comprises $n$ nodes ($n \geq 2$), communicating solely with their immediate neighbors, and remains static over time. 
In $\mathcal{G}$, the node set is represented as $\mathcal{V} = \{ v_1, v_2, ..., v_n \}$, and the edge set as $\mathcal{E} \subseteq \{ \mathcal{V} \times \mathcal{V} \} \cup \{ (v_i, v_i) \ | \ v_i \in \mathcal{V} \}$ (each node has a virtual self-edge). 
The number of nodes and edges are denoted by $| \mathcal{V} | = n$ and $| \mathcal{E} | = m$, respectively. 
A directed edge from node $v_i$ to node $v_l$ is denoted by $(v_l, v_i) \in \mathcal{E}$, signifying that node $v_l$ can receive information from node $v_i$ at time step $k$ (but not vice versa). 
The in-neighbors of $v_i$ are the nodes that can directly transmit information to $v_i$, denoted by $\mathcal{N}_i^- = \{ v_j \in \mathcal{V} ; | ; (v_i, v_j)\in \mathcal{E}\}$. 
The out-neighbors of $v_i$ are the nodes that can directly receive information from $v_i$, represented by $\mathcal{N}_i^+ = \{ v_l \in \mathcal{V} ; | ; (v_l, v_i)\in \mathcal{E}\}$. 
The in-degree and out-degree of $v_j$ are denoted by $\mathcal{D}_i^- = | \mathcal{N}_i^- |$ and $\mathcal{D}_i^+ = | \mathcal{N}_i^+ |$, respectively. 
The diameter $D$ of a digraph is the longest shortest path between any two nodes $v_l, v_i \in \mathcal{V}$. 
A directed path from $v_i$ to $v_l$ of length $t$ exists if there's a sequence of nodes $i \equiv l_0,l_1, \dots, l_t \equiv l$ such that $(l{\tau+1}, l{\tau}) \in \mathcal{E}$ for $ \tau = 0, 1, \dots , t-1$. 
A digraph is strongly connected if there's a directed path from every node $v_i$ to every node $v_l$ for all $v_i, v_l \in \mathcal{V}$.

\textbf{Node Operation.} 
During the operation of our proposed distributed optimization algorithm each node $v_i \in \mathcal{V}$ executes a local optimization step and then, a distributed coordination algorithm. 
In the distributed optimization algorithm, at each time step $k$, each node $v_i$ maintains (i) its local estimate variable $x_i^{[k]} \in \mathbb{Q}$ for computing the optimal solution, (ii) a learning rate $\alpha \in \mathbb{R}$ for optimizing its local function, (iii) the quantizer basis $b_q \in \mathbb{Q}$ that can be used for shifting and focusing on different ranges of values, (iv) the quantizer level $\Delta^{[\nu_{\text{total}}]} \in \mathbb{Q}$ that can be quantizing its local estimate variable ($\nu_{\text{total}}$ is explained below), 
(v) a zoom-in constant $c^{\text{in}} \in \mathbb{Q}_{> 1}$ for reducing the quantization level in order to estimate the optimal solution more precisely, (vi) a zoom-in counter $\nu_{\text{in}} \in \mathbb{N}$ in order to count the times nodes performed zoom-in, (vii) a zoom-out constant $c^{\text{out}} \in \mathbb{Q}_{> 1}$ for increasing the quantization level so that the quantizer range covers a wider range of values that may include the optimal solution, (viii) a zoom-out counter $\nu_{\text{out}} \in \mathbb{N}$ in order to count the times nodes performed zoom-out, and (ix) a total zoom counter $\nu_{\text{total}} \in \mathbb{N}$ in order to count the times nodes performed zoom operation (note that $\nu_{\text{total}} = \nu_{\text{in}} + \nu_{\text{out}}$). 
In the coordination algorithm, at each time step $k$, every node $v_i$ maintains (i) stopping variables $M_i$ and $m_i \in \mathbb{N}$ to determine convergence, (ii) mass variables $y_i^{[\lambda]}, z_i^{[\lambda]} \in \mathbb{Q}$ for communication by sending or receiving messages with other nodes, and (iii) state variables $y_{i,(s)}^{[\lambda]}, z_{i,(s)}^{[\lambda]}, q_{i,(s)}^{[\lambda]} \in \mathbb{Q}$ to store received messages and compute the coordination algorithm's result.

\textbf{$3$-bit Mid-Rise Uniform Quantizer with Base Shifting.} 
A quantizer discretizes a continuous range of values into a finite set of discrete values, reducing the number of bits needed to represent information. 
This compression minimizes bandwidth requirements for message transmission and enhances power and computational efficiency. 
Quantization is vital in scenarios like wireless communication, distributed control systems, and sensor networks, where communication constraints and imperfect information exchanges occur. 
The most common quantizers are uniform quantizers that partition the continuous value range into equally-sized intervals. 
Other types of quantizers include uniform midtread and non-uniform. 
For the results of this paper, we will utilize a $3$-bit mid-rise uniform quantizer with base shifting capabilities, defined as follows: 
\begin{align*}
&Q^{\text{$3$MRU}}(b_q, \xi, \Delta^{[\nu_{\text{total}}]})  \\
    &=\left\{ \begin{array}{ll}
         b_q - \dfrac{7 \Delta^{[\nu_{\text{total}}]}}{2}, & \xi \in (-\infty, b_q-3\Delta^{[\nu_{\text{total}}]}) \\ \ \vspace{-.3cm} \\ 
         b_q - \dfrac{5 \Delta^{[\nu_{\text{total}}]}}{2}, & \xi \in [b_q-3\Delta^{[\nu_{\text{total}}]}, b_q-2\Delta^{[\nu_{\text{total}}]}) \\ \ \vspace{-.3cm} \\ 
         b_q - \dfrac{3 \Delta^{[\nu_{\text{total}}]}}{2}, & \xi \in [b_q-2\Delta^{[\nu_{\text{total}}]}, b_q-\Delta^{[\nu_{\text{total}}]}) \\ \ \vspace{-.3cm} \\ 
         b_q - \dfrac{\Delta^{[\nu_{\text{total}}]}}{2}, & \xi \in [b_q-\Delta^{[\nu_{\text{total}}]}, b_q) \\ \ \vspace{-.3cm} \\ 
         b_q + \dfrac{\Delta^{[\nu_{\text{total}}]}}{2}, & \xi \in [b_q, b_q + \Delta^{[\nu_{\text{total}}]}) \\ \ \vspace{-.3cm} \\ 
         b_q + \dfrac{3 \Delta^{[\nu_{\text{total}}]}}{2}, & \xi \in [b_q + \Delta^{[\nu_{\text{total}}]}, b_q + 2\Delta^{[\nu_{\text{total}}]}) \\ \ \vspace{-.3cm} \\ 
         b_q + \dfrac{5 \Delta^{[\nu_{\text{total}}]}}{2}, & \xi \in [b_q + 2 \Delta^{[\nu_{\text{total}}]}, b_q + 3 \Delta^{[\nu_{\text{total}}]}) \\ \ \vspace{-.3cm} \\
         b_q + \dfrac{7 \Delta^{[\nu_{\text{total}}]}}{2}, & \xi \in [b_q + 3 \Delta^{[\nu_{\text{total}}]}, + \infty) \\ \ \vspace{-.3cm} \\
    \end{array} \right.
\end{align*} 
where $\xi \in \mathbb{R}$ is the value to be quantized, $b_q \in \mathbb{Q}$ is the basis of the quantizer,  $Q^{\text{$3$MRU}}(b_q, \xi,\Delta^{[\nu_{\text{total}}]})$ is the quantized version of $\xi$, and $\Delta^{[\nu_{\text{total}}]} \in \mathbb{Q}$ is the quantization level. 
Our $3$-bit mid-rise uniform quantizer is also shown in Fig.~\ref{MRU_model_fig}. 
Note here that our findings can also be generalized to other types of quantizers. 

\begin{figure}[ht]
\centering
\begin{tikzpicture}[scale=0.69,transform shape,
    thick,
    >=stealth',
     empty dot/.style = { circle, draw, fill = white!0,
                          inner sep = 0pt, minimum size = 4pt },
    filled dot/.style = { empty dot, fill = black}
  ]
  \def\r{2}
  \draw[->] (-6,0) -- (6,0) coordinate[label = {below:$\xi$}] (xmax);
  \draw[densely dashed, draw=lightgray] (0,-3.75) -- (0,3.75) coordinate[label = {left:$Q^{3\mathrm{MRU}}$}]  (ymax);
 \draw [draw=blue] (3,3.5) -- (3,2.5);
\node [filled dot, draw=blue, fill=blue] at (3,3.5) {};
\draw  [draw=blue] (3,3.5) -- (5,3.5); 
  \foreach \i in {\r+1,\r,\r-1} {
    \draw [draw=blue] (\i-1,\i-1/2)   -- (\i,\i-1/2);
    \draw [draw=blue] (\i-1,\i-1/2)   -- (\i-1,\i-3/2);
    \draw [densely dashed, draw=lightgray] (\i,\i-1/2) -- (\i,0);
    \node [filled dot, draw=blue, fill=blue] at (\i-1,\i-1/2) {};
    \node [empty  dot,draw=blue] at (\i,\i-1/2) {};
  }

   \draw[draw=blue] (-3,-3.5) -- (-5,-3.5);
  
  \foreach \i in {\r-5,\r-4,\r-3} {
    \draw [draw=blue] (\i,\i+1/2)   -- (\i+1,\i+1/2);
    \draw [draw=blue] (\i,\i+1/2)  -- (\i,\i-1/2);
    \draw [densely dashed, draw=lightgray] (\i,\i+1/2) -- (\i,0);
    \node [filled dot, draw=blue, fill=blue] at (\i,\i+1/2) {};
  }
  \foreach \i in {\r-6,\r-5,\r-4,\r-3} {
  \node [empty  dot,draw=blue] at (\i+1,\i+1/2) {};
 }

  \node ["right:\scriptsize $b_q$"] at (\r-2.2,-0.2) {};
  \node ["below:\scriptsize $b_q\!+\!\bar{\Delta}$"]     at (\r-1,0.18)   {};
  \node ["below:\scriptsize $b_q\!+\!2\bar{\Delta}$"]     at (\r,0.18)   {};
  \node ["below:\scriptsize $b_q\!+\!3\bar{\Delta}$"]     at (\r+1,0.18)   {};
    
  \node ["above:\scriptsize $b_q\!-\!\bar{\Delta}$"]     at (\r-3,-0.18)   {};  
  \node ["above:\scriptsize $b_q\!-\!2\bar{\Delta}$"]     at (\r-4,-0.18)   {}; 
  \node ["above:\scriptsize $b_q\!-\!3\bar{\Delta}$"]     at (\r-5,-0.18)   {};  
  
  \node ["below:\scriptsize $000$"]     at (\r-5.5,-3.32)   {};  
  \node ["below:\scriptsize $001$"]     at (\r-4.5,-2.32)   {};  
  \node ["below:\scriptsize $010$"]     at (\r-3.5,-1.32)   {};  
   \node ["below:\scriptsize $011$"]     at (\r-2.5,-0.32)   {}; 
   
   \node ["above:\scriptsize $100$"]     at (\r-1.5,0.32)   {}; 
   \node ["above:\scriptsize $101$"]     at (\r-0.5,1.32)   {}; 
   \node ["above:\scriptsize $110$"]     at (\r+0.5,2.32)   {}; 
    \node ["above:\scriptsize $111$"]     at (\r+1.5,3.32)   {}; 
\end{tikzpicture}
\caption{A $3$-bit mid-rise uniform quantizer,  $Q^{\text{$3$MRU}}(b_q, \xi,\bar{\Delta})$.}
\label{MRU_model_fig}
\end{figure}
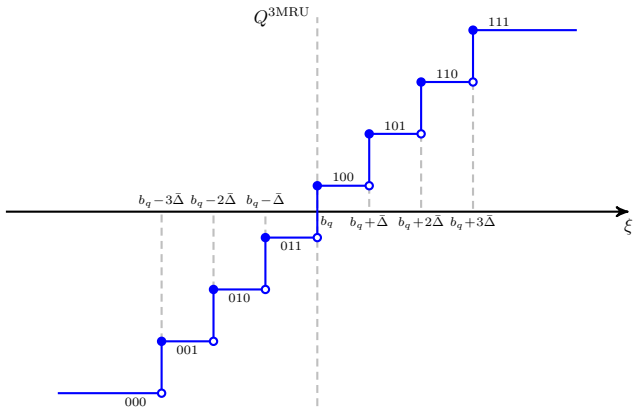

Our quantizer (which is specified for exposition purposes only) utilizes $3$ bits to represent quantized values within the output range $\{b_q - \dfrac{7 \Delta^{[\nu_{\text{total}}]}}{2}, b_q - \dfrac{5 \Delta^{[\nu_{\text{total}}]}}{2}, ..., b_q + \dfrac{7 \Delta^{[\nu_{\text{total}}]}}{2}\}$, and has $2^3$ output levels (i.e., the output range contains $2^3$ elements). 
Furthermore, note that our quantizer saturates for input values in the range $\xi \in (-\infty, b_q - 3\Delta^{[\nu_{\text{total}}]}) \cup [b_q + 3 \Delta^{[\nu_{\text{total}}]}, +\infty)$. 
This means that the input signal $\xi$ exceeds our quantizer's dynamic range, leading to clipping of the quantized output value to the nearest extreme value within the quantization range. 
This characteristic highlights the importance of selecting an appropriate quantization range to ensure an accurate representation of input signals without distortion due to saturation effects. 


%
%
%
%
\section{Problem Formulation}\label{sec:probForm}

Let us consider a distributed network modeled as a digraph $\mathcal{G} = (\mathcal{V}, \mathcal{E})$ with $n  = | \mathcal{V} |$ nodes. 
We assume that each node $v_i$ is endowed with a local cost function $f_i(x): \mathbb{R}^p \mapsto \mathbb{R}$ only known to itself, and communication channels among nodes have limited capacity of $3$ bits per channel, thus, only quantized values can be transmitted/communicated.
In this paper we aim to develop a distributed algorithm which allows nodes, to cooperatively solve the following optimization problem called \textbf{P1} below: 
\begin{subequations}
\begin{align}
\min_{x\in \mathcal{X}}~ & F(x_1, x_2, ..., x_n) \equiv \sum_{i=1}^n f_i(x_i), \label{Global_cost_function}  \\
\text{s.t.}~ & x_i = x_j, \forall v_i, v_j, \in \mathcal{V}, \label{constr_same_x}  \\
       & x_i^{[0]} \in \mathcal{X} \subset \mathbb{Q}_{\geq 0}, \forall v_i \in \mathcal{V}, \label{constr_x_in_X} \\
       & \text{nodes communicate via $3$ bit channels, } \label{constr_quant}  
\end{align} 
\end{subequations} 
where $\mathcal{X}$ is the set of feasible values of variable $x$, and $x^*$ is the optimal solution of the optimization problem.  
Eq.~\eqref{Global_cost_function} means that our objective is to minimize the global cost function, defined as the sum of the local cost functions in the network. 
Eq.~\eqref{constr_same_x} imposes the requirement that nodes need to calculate equal optimal solutions. 
Eq.~\eqref{constr_x_in_X} means that the initial estimations of nodes must belong in the same set. 
However, it is not necessary for the initial values of nodes to be rational numbers (i.e., $x_i^{[0]} \in \mathcal{X} \subset \mathbb{Q}_{\geq 0}$), as they can generate a quantized version of their initial states by utilizing the Mid-Rise Uniform Quantizer presented in Section~\ref{sec:preliminaries}. 
Eq.~\eqref{constr_quant} means that nodes are transmitting and receiving $3$ bit quantized values with their neighbors. 
This constraint reflects the limited bandwidth of the network's communication channels. 


\begin{remark}[Challenges of Optimization Problem]   
The main challenges of our optimization problem \textbf{P1} are due to constraint \eqref{constr_quant} (i.e., nodes are limited to exchanging messages represented by only $3$ bits). 
Our proposed algorithm is able to address these challenges by enabling nodes to decide whether to zoom-in or zoom-out to the optimal solution, ensuring both communication efficiency and avoidance of saturation in the quantizer output range. 
Moreover, our algorithm enables nodes to calculate the exact optimal solution while exchanging $3$-bit messages, mitigating quantization errors and saturation issues.
\end{remark}

%
%
%
%
\section{Distributed Optimization with Finite Bit Adaptive Quantization}\label{sec:distr_grad_zoom_quant}

In this section, we present our distributed optimization algorithm which solves problem \textbf{P1} described in Section~\ref{sec:probForm}. 
Before presenting our algorithm, we make the following assumptions that are crucial for the development of our results. 


\begin{assum}\label{str_conn}
The communication network described as a digraph $\mathcal{G}$ is \textit{strongly connected}. 
\end{assum}

\begin{assum}\label{assup_convex}
    The local cost function $f_i(x)$ of every node $v_i$ is smooth and strongly convex. 
    This means that for every node $v_i$, for every $x_1, x_2 \in \mathcal{X}$, 
    \begin{list4}
        \item there exists positive constant $L_i$ such that  \begin{equation}\label{lipschitz_defn}
            \| \nabla f_i(x_1) - \nabla f_i(x_2) \|_2 \leq L_i \| x_1 - x_2 \|_2, 
        \end{equation}
        \item there exists positive constant $\mu_i$ such that 
        \begin{equation}\label{str_conv_defn}
             f_i(x_2) \geq f_i(x_1) + \nabla f_i(x_1)^\top (x_2 - x_1) + \frac{\mu_i}{2} \| x_2 - x_1 \|_2^2. 
        \end{equation} 
    \end{list4} 
\noindent This means that for the global cost function $F$ (see \eqref{Global_cost_function}), the Lipschitz-continuity constant of the gradient is $L =  \sum_i L_i$, and the strong-convexity constant is $\mu = \min_i \mu_i$. 
\end{assum}

\begin{assum}\label{digr_diam}
The diameter $D$ (or an upper bound of the diameter) of the digraph $\mathcal{G}$ is known to every node $v_i$ in the network. 
\end{assum}

Assumption~\ref{str_conn} ensures that information from each node propagates to every other node in the network.
Thus all nodes are able to calculate the optimal solution $x^*$ of $P1$. 
Assumption~\ref{assup_convex} is the Lipschitz-continuity condition \eqref{lipschitz_defn}, and strong-convexity condition \eqref{str_conv_defn}. 
Lipschitz-continuity is a standard assumption in distributed first-order optimization problems (see \cite{2018:Xie, 2018:Li_Quannan}). 
It guarantees that problem $P1$ has an optimal solution $x^*$ for nodes to calculate, and that this solution is globally optimal. 
Strong-convexity guarantees linear convergence rate, and uniqueness of the optimal solution (i.e., the global function $F$ has no more than one global minimum). 
Assumption~\ref{digr_diam} enables each node $v_i \in \mathcal{V}$ to determine whether calculation of a solution $x_i$ that fulfills \eqref{constr_same_x} has been achieved in a distributed manner. 
Note here that nodes can distributively compute the digraph diameter $D$ in finite time by employing a graph diameter calculation protocol such as \cite{oliva2016distributed}.

\subsection{Distributed Optimization Algorithm with Finite Bit Adaptive Quantization}\label{distr_alg}

We now present our proposed algorithm detailed below as Algorithm~\ref{alg1}. 

\begin{varalgorithm}{1}
\caption{Distributed Optimization Algorithm with Finite Bit Adaptive Quantization (DOFiBAQ)}
\textbf{Input:} A strongly connected directed graph $\mathcal{G}$ with $n = |\mathcal{V}|$ nodes and $m = |\mathcal{E}|$ edges. 
Every node $v_i \in \mathcal{V}$ has: static step-size $\alpha \in \mathbb{R}$, digraph diameter $D$, initial value $x_i^{[0]}$, local cost function $f_i$, quantization level $\Delta^{[0]} \in \mathbb{Q}$, zoom-in constant $c^{\text{in}} \in \mathbb{Q}_{> 1}$, zoom-out constant $c^{\text{out}} \in \mathbb{Q}_{> 1}$. 
Assumptions~\ref{str_conn}, \ref{assup_convex}, \ref{digr_diam} hold. 
\\
\textbf{Initialization:} Each node $v_i \in \mathcal{V}$ sets $b_q = 0$, $\nu_{\text{in}} = 0$, $\nu_{\text{out}} = 0$, $\nu_{\text{total}} = 0$. \\ 
\textbf{Iteration:} For $k = 0,1,2, \dots$, each node $v_i \in \mathcal{V}$ does the following: 
\begin{list4}
\item[1)] $x_i^{[k+\frac{1}{2}]} =  x_i^{[k]} - \alpha \nabla f_i(x_i^{[k]})$; 
\item[2)] $x_i^{[k+1]} = $ Algorithm~\ref{alg2}($x_i^{[k+\frac{1}{2}]}, D, \Delta^{[\nu_{\text{total}}]} $); 
\item[3)] \textbf{if} $x_i^{[k+1]} = x_i^{[k]}$, \textbf{then} set $\nu_{\text{total}} = \nu_{\text{total}} + 1$ and \textbf{check}
\begin{list4a}
\item \textbf{if} $x_i^{[k+1]} \geq b_q + 3 \Delta^{[\nu_{\text{total}}]}$, \textbf{or} $x_i^{[k+1]} < b_q - 3 \Delta^{[\nu_{\text{total}}]}$ \textbf{then} set 
\begin{subequations} 
\begin{align} 
    & \nu_{\text{out}} = \nu_{\text{out}} + 1, \label{zoom_out_0} \\ 
    & b_q = x_i^{[k+1]},\label{zoom_out_1} \\
    & \Delta^{[\nu_{\text{total}}]} := c^{\text{out}} \Delta^{[\nu_{\text{total}}-1]},\label{zoom_out_2} 
\end{align} 
\end{subequations} 
\item \textbf{else} set 
\begin{subequations} 
\begin{align} 
    & \nu_{\text{in}} = \nu_{\text{in}} + 1, \label{zoom_in_0} \\ 
    & b_q = x_i^{[k+1]},\label{zoom_in_1} \\
    & \Delta^{[\nu_{\text{total}}]} := \Delta^{[\nu_{\text{total}}-1]} / c^{\text{in}},\label{zoom_in_2} 
\end{align} 
\end{subequations} 
\end{list4a} 
\item[$4)$] go to Step~$1$;
\end{list4} 
\textbf{Output:} Each node $v_i \in \mathcal{V}$ calculates the optimal solution $x_i^*$ of problem \textbf{P1} in Section~\ref{sec:probForm}. 
\label{alg1} 
\end{varalgorithm}

\noindent
\vspace{-0.3cm}    
\begin{varalgorithm}{2}
\caption{Finite-Time Quantized Average Consensus (FiTQuAC)}
\textbf{Input:} $x_i^{[k+\frac{1}{2}]}, D, \Delta^{[\nu_{\text{total}}]}$. 
\\
\textbf{Initialization:} Each node $v_i \in \mathcal{V}$ does the following: 
\begin{list4}
\item[$1)$] Assigns probability $b_{li}$ to each out-neighbor $v_l \in \mathcal{N}^+_i \cup \{v_i\}$, as follows
\begin{align*}
b_{li} = \left\{ \begin{array}{ll}
         \frac{1}{1 + \mathcal{D}_i^+}, & \mbox{if $l = i$ or $v_{l} \in \mathcal{N}_i^+$,} \\
         0, & \mbox{if $l \neq i$ and $v_{l} \notin \mathcal{N}_i^+$;}\end{array} \right. 
\end{align*} 
\item[$2)$] sets $z_i = 2$, $y_i = 2 \  Q^{\text{$3$MRU}}(b_q, x_i^{[k+\frac{1}{2}]}, \Delta^{[\nu_{\text{total}}]}) / \Delta^{[\nu_{\text{total}}]}$.  
\end{list4} 
\textbf{Iteration:} For $\lambda = 1,2,\dots$, each node $v_i \in \mathcal{V}$, does: 
\begin{list4a}
\item[$1)$] \textbf{if} $\lambda \mod (D) = 1$ \textbf{then} $M_i = \lceil y_i  / z_i \rceil$, $m_i = \lfloor y_i / z_i \rfloor$; 
\item[$2)$] broadcasts $M_i$, $m_i$ to every $v_{l} \in \mathcal{N}_i^+$; receives $M_j$, $m_j$ from every $v_{j} \in \mathcal{N}_i^-$; sets $M_i = \max_{v_{j} \in \mathcal{N}_i^-\cup \{ v_i \}} M_j$, \\ $m_i = \min_{v_{j} \in \mathcal{N}_i^-\cup \{ v_i \}} m_j$; 
\item[$3)$] sets $c_i^z = z_i$; 
\item[$4)$] \textbf{while} $c_i^z > 1$ \textbf{do} 
\begin{list4a}
\item[$4.1)$] $c^y_i = \lfloor y_{i} \  / \  z_{i} \rfloor$; 
\item[$4.2)$] sets $y_{i} = y_{i} - c^y_i$, $z_{i} = z_{i} - 1$, and $c_i^z = c_i^z - 1$; 
\item[$4.3)$] transmits $c^y_i$ to randomly chosen out-neighbor $v_l \in \mathcal{N}^+_i \cup \{v_i\}$ according to $b_{li}$; 
\item[$4.4)$] receives $c^y_j$ from $v_j \in \mathcal{N}_i^-$ and sets 
\begin{align}
y_i & = y_i + \sum_{j=1}^{n} w^{[r]}_{\lambda,ij} \ c^y_{j} \ , \\
z_i & = z_i + \sum_{j=1}^{n} w^{[r]}_{\lambda,ij} \ ,
\end{align}
where $w^{[r]}_{\lambda,ij} = 1$ when node $v_i$ receives $c^y_{i}$, $1$ from $v_j$ at time step $\lambda$ (otherwise $w^{[r]}_{\lambda,ij} = 0$ and $v_i$ receives no message at time step $\lambda$ from $v_j$);
\end{list4a} 
\item[$5)$] \textbf{if} $\lambda \mod D = 0$ \textbf{and} $M_i - m_i \leq 1$ \textbf{then} sets $x_i^{[k+1]} = m_i \Delta^{[\nu_{\text{total}}]}$ and stops operation. 
\end{list4a}
\textbf{Output:} $x_i^{[k+1]}$. 
\label{alg2} 
\end{varalgorithm}
 
The intuition of Algorithm~\ref{alg1} (DOFiBAQ) can be summarized in the following three parts. 
\\ \noindent
\textit{Part~$1$: Input and Initialization.} 
Each node $v_i$ retains an estimate $x_i^{[0]}$ of the optimal solution. 
Also, it has knowledge of (i) the network diameter $D$, (ii) a quantization level $\Delta^{[0]}$, (iii) a zoom-in constant $c^{\text{in}} \in \mathbb{Q}_{> 1}$, (iv) a zoom-out constant $c^{\text{out}} \in \mathbb{Q}_{> 1}$, and (v) a static step-size $\alpha \in \mathbb{R}$. 
The basis of the quantizer $b_q$ (see Section~\ref{sec:preliminaries}) is initialized to be equal to zero (our algorithm remains functional also for different arbitrary initial values of $b_q$). 
Also, the total zoom, zoom-in and zoom-out counters $\nu_{\text{total}}$, $\nu_{\text{in}}$, $ \nu_{\text{out}}$ are initialized equal to zero. 
The aforementioned parameters are common for every node. 
Note however that nodes can decide the parameter values in finite time by executing either a voting protocol \cite{2008:Cortes}, or a digraph diameter calculation protocol \cite{oliva2016distributed} (as mentioned in Assumption~\ref{digr_diam}). 
\\ \noindent
\textit{Part~$2$: Iteration -- Gradient Descent and Coordination.} 
During each time step $k$ of Algorithm~\ref{alg1}, each node $v_i$ updates the estimation $x_i^{[k+\frac{1}{2}]}$ of its optimal solution by performing a gradient descent step towards the negative direction of the gradient of the node's local cost function. 
Then, each node employs Algorithm~\ref{alg2} (FiTQuAC). 
Algorithm~\ref{alg2} is a finite-time quantized averaging algorithm (details of its operation are presented below). 
This algorithm enables nodes to satisfy \eqref{constr_quant} and compute in finite time an estimation of the optimal solution $x_i^{[k+1]}$ that satisfies \eqref{constr_same_x}.
\\ \noindent 
\textit{Part~$3$: Iteration -- Zoom-out or Zoom-in.} 
After executing Algorithm~\ref{alg2}, nodes check if the calculated estimation of the optimal solution is the same as the previous step (i.e., if $x_i^{[k+1]} = x_i^{[k]}$). 
The condition $x_i^{[k+1]} = x_i^{[k]}$ means that nodes have converged to a neighborhood of the optimal solution defined by the current quantization level $\Delta^{[\nu_{\text{total}}]}$ (see \cite{2023:Rikos_Johan_IFAC, rikos2023distributed}). 
Additionally, when this condition is met nodes are not able to enhance the precision of the estimated optimal solution using the existing quantization level. 
For this reason, they need to adjust the parameters of the utilized quantizer. 
Therefore, if the aforementioned condition holds, then nodes increase by one the total zoom counter $\nu_{\text{total}}$ check: 
\begin{itemize}
    \item \textit{Zoom-out:} 
    If the estimation exceeds the quantizer's dynamic range (i.e., if the quantizer has saturated because $x_i^{[k+1]} \in (-\infty, b_q - 3\Delta^{[\nu_{\text{total}}]}) \cup [b_q + 3 \Delta^{[\nu_{\text{total}}]}, +\infty)$), then the optimal solution of problem \textbf{P1} falls outside the current quantizer range. 
    This means that nodes are not able to estimate the optimal solution with the current quantizer parameters. 
    In this case nodes zoom-out. 
    This is done by (i) setting the quantizer basis to be equal to the calculated estimation $x_i^{[k+1]}$, and (ii) increasing the quantization level by multiplying it with the zoom-out constant. 
    Also, nodes increase by one the zoom-out counter  $\nu_{\text{out}}$ in order to keep track how many times they performed zoom-out. 
    The primary purpose of zoom-out is to adapt the quantizer's parameters to encompass a broader range of values that may potentially include the optimal solution. 
    \item \textit{Zoom-in:} 
    If the estimate falls within the quantizer's dynamic range (i.e., if the quantizer is not saturated and we have $x_i^{[k+1]} \in [b_q - 3\Delta^{[\nu_{\text{total}}]}, b_q + 3 \Delta^{[\nu_{\text{total}}]})$), then nodes are able to estimate the optimal solution of problem \textbf{P1} with the current quantizer parameters. 
    In this case nodes zoom-in by (i) setting the quantizer basis to be equal to the calculated estimation $x_i^{[k+1]}$, and (ii) decreasing the quantization level by dividing it with the zoom-in constant. 
    Also, nodes increase by one the zoom-in counter  $\nu_{\text{in}}$ to count the amount of zoom-in operations.  
    By zoom-in nodes aim to adapt the quantizer's parameters to 
    ``focus'' at a a smaller range of values that may potentially include the optimal solution. 
    In this smaller range of values nodes are able to estimate with higher precision (since the updated $\Delta^{[\nu_{\text{total}}]}$ is smaller than the original). 
    Later in the paper, we will demonstrate how nodes are able to estimate the exact optimal solution by exchanging only $3$-bit quantized messages with their neighbor nodes (thus fulfilling constraint \eqref{constr_quant}).
\end{itemize}

Algorithm~\ref{alg2} (FiTQuAC) enables each node to calculate the quantized average of each node’s estimate in finite time. 
Additionally it incorporates a distributed stopping mechanism to notify nodes when calculation has converged so that nodes can return to Iteration Step~$3$ of Algorithm~\ref{alg1}. 
During Algorithm~\ref{alg2}, each node processes and transmits quantized messages, and the precision of the average calculation is determined by the quantizer type (i.e., its quantization level, and its saturation limits). 
An analytical intuition of Algorithm~\ref{alg2} (FiTQuAC) is shown in \cite[Algorithm~$2$]{rikos2023distributed}. 
We omit it due to space constraints.

\subsection{Convergence of Algorithm~\ref{alg1}}\label{Conv_Alg}

We now analyze the convergence of Algorithm~\ref{alg1}. 
We first introduce the following proposition to show that by zoom out nodes are able to adjust the quantizer dynamic range in finite time so it includes the optimal solution.  
We then present a theorem which establishes the linear convergence rate of our algorithm. 
Due to space constraints, we omit the proofs. 
They will be available in an extended version of our paper. 

\begin{prop}[Finite Zoom-out Instances]\label{zoom_out_locate_solution}   
    Let us assume that during the operation of Algorithm~\ref{alg1}, $x^* \notin (b_q - 3 \Delta^{[0]}, b_q + 3\Delta^{[0]})$. 
    Then, after a finite number of time steps $\nu_0$ for which 
    \begin{equation}\label{bound_time_steps}
        \nu_0 > \left\lceil \frac{x^* - \log(3\Delta)}{\log(c^{\text{out}})} \right\rceil , 
    \end{equation}
    nodes are able to calculate $\Delta^{[\nu_0]}$ such that $x^* \in (b_q - 3 \Delta^{[\nu_0]}, b_q + 3\Delta^{[\nu_0]})$. 
\end{prop}

We now show the linear convergence rate of Algorithm~\ref{alg1} in the following theorem. 
The proof follows from \cite[Theorem~1]{2023:Rikos_Johan_IFAC} and will be available in an extended version of our paper. 

\begin{theorem}[Linear Convergence Rate \cite{2023:Rikos_Johan_IFAC}]\label{converge_Alg1}
Under Assumptions~\ref{str_conn}--\ref{digr_diam}, when the step-size $\alpha$ satisfies $ \alpha \in (0, \frac{2n}{\mu + L}] $ where 
$L =  \sum_i L_i$, and $\mu = \min_i \mu_i$ 
Algorithm~\ref{alg1} generates a sequence of points $ \{x^{[k]}\} $ (i.e., the variable $x_i^{[k]}$ of each node $v_i \in \mathcal{V}$) which satisfy
	\begin{align}\label{linear_convergence}
	\| \hat{x}^{[k+1]} - x^{*}\|
	\le
	(1-\frac{\alpha \mu}{n})\| \hat{x}^{[k]} - x^{*}  \| + \mathcal{O}(\Delta^{[\nu_{\text{total}}]}) ,
	\end{align}
	where $\Delta^{[\nu_{\text{total}}]} $ is the quantizer level and
		\begin{align}
		 \mathcal{O}(\Delta^{[\nu_{\text{total}}]}) 
		 =&( \frac{4\alpha L}{n} + 2)\Delta^{[\nu_{\text{total}}]}. \label{throrem1_b} 
		\end{align} 
\end{theorem} 

\section{Simulation Results} \label{sec:results}





In this section, we present simulation results in order to illustrate Algorithm~\ref{alg1} and highlight its operational characteristics. 
We illustrate Algorithm~\ref{alg1} over a random digraph of $20$ nodes (see Fig.~\ref{20_nodes_simple_finite_bit}).   
We show how nodes are able to adjust the quantizer’s parameters by performing zoom-out, or zoom-in. 
Subsequently, we show how nodes are able to converge linearly to the exact optimal solution of problem \textbf{P1} by exchanging messages consisted of only $3$-bits.   

During Algorithm~\ref{alg1} the parameters for each node $v_i$ are: $\alpha = 0.12$, initial estimation of optimal solution $x_i^{[0]} \in [1, 5]$, quadratic local cost function of the form $f_i(x) = \frac{1}{2} \beta_i (x - x_0)^2$ with $\beta_i, x_0$ randomly chosen in the set $\{ 1, 2, 3, 4, 5 \}$, initial quantization level $\Delta^{[0]} = 0.5$, zoom-in constant $c^{\text{in}} = 4/3$, zoom-out constant $c^{\text{out}} = 2$, and initial quantizer basis $b_q = 0$. 
We plot the error $e^{[k]}$ in a logarithmic scale against the number of iterations defined as 
\begin{equation}\label{eq:distance_optimal}
    e^{[k]} = \sqrt{ \sum_{j=1}^n \frac{(x_j^{[k]} - x^*)^2}{(x_j^{[0]} - x^*)^2} } , 
\end{equation}
where $x^*$ is the optimal solution of optimization problem \textbf{P1}.   

\vspace{.3cm}

In Fig.~\ref{20_nodes_simple_finite_bit} during the operation of Algorithm~\ref{alg1} we can see that at time step $k=2$ the condition $x_i^{[2]} = x_i^{[1]}$ holds for all nodes. 
Thus, nodes perform zoom-out during time step $k = 2$. 
They adjust the quantizer basis $b_q$, and increase the quantization level $\Delta^{[0]}$ according to \eqref{zoom_out_1} and \eqref{zoom_out_2}, thereby widening the quantizer range to encompass a broader range of values. 
At time step $k=4$ the condition $x_i^{[4]} = x_i^{[3]}$ holds. 
Nodes again perform zoom out (i.e. adjust the quantizer parameters $b_q, \Delta^{[1]}$) aiming for the adjusted quantizer range to include the optimal solution. 
Since the quantizer error depends on the quantization level then increasing $\Delta^{[\nu_{\text{total}}]}$ increases the error introduced in the estimation of the optimal solution. 
This increases the distance between the optimal solution and each node's estimation of the optimal solution. 
For this reason, the error $e^{[k]}$ increases at time steps $k = 3, 5$. 
At time step $k = 6$ we have that $x_i^{[6]} = x_i^{[5]}$ and nodes decide to zoom-in. 
Performing zoom-in at time step $k = 6$ means that the zoom-out operations at time steps $k = 2, 4$ adjusted $b_q, \Delta^{[0]}$ in such a way that the quantizer range now includes the optimal solution. 
During the zoom-in at time step $k = 6$, nodes adjust $b_q$, and decrease $\Delta^{[2]}$ according to \eqref{zoom_in_1} \eqref{zoom_in_2}. 
This enables them to approximate the optimal solution with higher precision at time step $k = 7$ (as we can see $e^{[7]} < e^{[6]}$). 
Then, at time step $k = 9$ nodes zoom in since $x_i^{[9]} = x_i^{[8]}$. 
This enables them to approximate the optimal solution with even higher precision at time step $k = 10$. 
By continuing this iterative process (i.e, perform zoom-in and approximate the optimal solution with even higher precision every time) we have that nodes are able to converge to the exact optimal solution during the operation of Algorithm~\ref{alg1} (as shown in Theorem~\ref{converge_Alg1}). 

\begin{figure}[t]
    \centering
    \includegraphics[width=\linewidth]{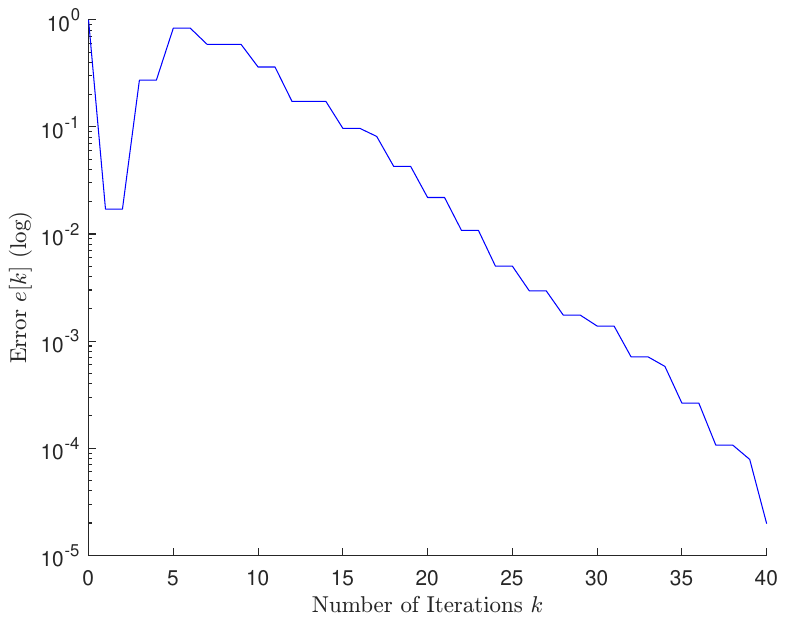}
    \caption{Execution of Algorithm~\ref{alg1} over a random digraph of $20$ nodes.}
    \label{20_nodes_simple_finite_bit}
\end{figure}


\begin{remark}[Possible $c^{\text{in}}$ Values]  
    Regarding the operation of Algorithm~\ref{alg1} in Fig.~\ref{20_nodes_simple_finite_bit} we make the following two observations for the zoom-in constant $c^{\text{in}}$.  
    \begin{itemize}
        \item In Algorithm~\ref{alg1}, nodes set $c^{\text{in}} = 4/3$. 
        This means that after the zoom-in operation the adjusted quantizer will cover the range of values $[b_q - 16\Delta^{[\nu_{\text{total}}-1]}/3 , b_q + 16\Delta^{[\nu_{\text{total}}-1]}/3)$. 
        Note that Algorithm~\ref{alg2} (i) introduces a quantization error upper bounded by $\Delta^{[\nu_{\text{total}}]} / 2$ during its initialization, and (ii) gives the average of nodes states with a quantization error $\Delta^{[\nu_{\text{total}}]}$. 
        Thus, the updated interval $[b_q - 16\Delta^{[\nu_{\text{total}}-1]}/3 , b_q + 16\Delta^{[\nu_{\text{total}}-1]}/3)$ includes the optimal solution despite the accumulated quantization error from Algorithm~\ref{alg2}. 
        Note also that setting $c^{\text{in}} > 2$ for our $3$ bit quantizer may lead to an adjusted quantizer range that does not include the optimal solution after reducing the quantization level $\Delta^{[\nu_{\text{total}}]}$. 
        This may lead to nodes performing zoom-out so that the updated range will include again the optimal solution. 
        As a result, setting $c^{\text{in}} > 2$ may lead to consecutive zoom-in and zoom-out operations that may prevent Algorithm~\ref{alg1} from converging to the exact optimal solution. 
        \item In case we utilize a quantizer with higher resolution (i.e., with $4$ or more bits) then setting the zoom-in constant equal to $c^{\text{in}} = 4/3$ will also allow us to zoom towards an interval that includes the optimal solution. 
        However, in this case we also could possibly increase the zooming constant to be greater than $4/3$ in order to zoom in faster thus enhance the speed of convergence of the Algorithm~\ref{alg1}. 
    \end{itemize}
    Note that analyzing a possible range of values for $c^{\text{in}}$ (according to the two aforementioned observations) is outside the scope of this paper but will be considered in an extended version. 
\end{remark}

\section{Conclusions}\label{sec:conclusions}

In this paper, we focused on distributed optimization problem. 
In our problem each node's local function is strongly convex with lipschitz continuous gradients, and nodes exchange messages represented by $3$-bits via bandwidth limited channels. 
We presented a distributed algorithm that enables nodes to estimate the exact optimal solution. 
We analyzed our algorithm's operation and established its linear convergence rate. 
We concluded with simulations 
in order to highlight our algorithm's operational advantages. 
\bibliographystyle{IEEEtran}
\bibliography{bibliografia_consensus}

\end{document}